\documentclass{llncs}
\usepackage{amsmath}
\usepackage{amssymb}
\usepackage{hyperref}

\usepackage{enumerate}
\usepackage{tikz}
\usepackage{boxedminipage}
\usepackage{microtype}
\usetikzlibrary{arrows,shapes,calc,positioning}
\hypersetup{hypertexnames=false}

\newcommand{\problemdef}[3]{
	\begin{center}
		\begin{boxedminipage}{.99\textwidth}
			\textsc{{#1}}\\[2pt]
			\begin{tabular}{ r p{0.8\textwidth}}
				\textit{~~~~Instance:} & {#2}\\
				\textit{Question:} & {#3}
			\end{tabular}
		\end{boxedminipage}
	\end{center}
}

\newcommand{\NP}{{\sf NP}}

\title{Independent Feedback Vertex Sets for\\ Graphs of Bounded Diameter\thanks{This paper received support from EPSRC (EP/K025090/1), London Mathematical Society (41536), the Leverhulme Trust (RPG-2016-258) and Fondation Sciences Math\'ematiques de Paris.
The hardness result (Theorem~\ref{thm:diam3}) of this paper has been announced in an extended abstract of the Proceedings of MFCS 2017~\cite{BDFJP17}.}}

\author{Marthe Bonamy\inst{1} \and
Konrad K. Dabrowski\inst{2} \and
Carl Feghali\inst{3}
\and\\ Matthew Johnson\inst{2}
\and Dani\"el Paulusma\inst{2}}
\institute{
CNRS, LaBRI, France \texttt{marthe.bonamy@u-bordeaux.fr},
\and
School of Engineering and Computing Sciences, Durham University, UK\\
\texttt{\{konrad.dabrowski,matthew.johnson2,daniel.paulusma\}@durham.ac.uk},
\and
IRIF \& Universit\'e Paris Diderot, Paris France \texttt{feghali@irif.fr}}

\newcommand\displaycase[1]{{\em #1}}

\newcounter{ctrclaim}[theorem]
\newcommand{\clm}[1]{\medskip\phantomsection\refstepcounter{ctrclaim}\noindent\displaycase{Claim \thectrclaim. }{\em #1}\\}

\begin{document}
\maketitle

\begin{abstract}
The {\sc Near-Bipartiteness} problem is that of deciding whether or not the vertices of a graph can be partitioned into sets~$A$ and~$B$, where~$A$ is an independent set and~$B$ induces a forest.
The set~$A$ in such a partition is said to be an independent feedback vertex set.
Yang and Yuan proved that {\sc Near-Bipartiteness} is polynomial-time solvable for graphs of diameter~$2$ and \NP-complete for graphs of diameter~$4$.
We show that {\sc Near-Bipartiteness} is \NP-complete for graphs of diameter~$3$, resolving their open problem.
We also generalise their result for diameter~$2$ by proving that even the problem of computing a minimum independent feedback vertex is polynomial-time solvable for graphs of diameter~$2$.

\keywords{near-bipartite graphs, independent feedback vertex set, diameter, computational complexity}
\end{abstract}

\section{Introduction}

A graph is {\em near-bipartite} if its vertex set can be partitioned into sets~$A$ and~$B$, where~$A$ is an independent set and~$B$ induces a forest.
The set~$A$ is said to be an {\em independent feedback vertex set} and the pair $(A,B)$ is said to be a {\em near-bipartite decomposition}.
This leads to the following two related decision problems.

\problemdef{\sc Near-Bipartiteness}{a graph~$G$.}{is~$G$ near-bipartite (that is, does~$G$ have an independent feedback vertex set)?}

\problemdef{\sc Independent Feedback Vertex Set}{a graph~$G$ and an integer $k\geq 0$.}{does~$G$ have an independent feedback vertex set of size at most~$k$?}

\noindent
Setting $k=n$ shows that the latter problem is more general than the first problem.
Thus, if {\sc Near-Bipartiteness} is \NP-complete for some graph class, then so is {\sc Independent Feedback Vertex Set}, and if {\sc Independent Feedback Vertex Set} is polynomial-time solvable for some graph class, then so is {\sc Near-Bipartiteness}.

\begin{sloppypar}
Note that every near-bipartite graph is $3$-colourable, that is, its vertices can be coloured with at most three colours such that no two adjacent vertices are coloured alike.
The problems {\sc $3$-Colouring}~\cite{Lo73} and {\sc Near-Bipartiteness}~\cite{BLS98} (and thus {\sc Independent Feedback Vertex Set}) are \NP-complete.
However, their complexities do not necessarily coincide on special graph classes.
Gr\"otschel, Lov\'asz and Schrijver~\cite{GLS84} proved that {\sc Colouring} is polynomial-time solvable for perfect graphs even if the permitted number of colours~$k$ is part of the input.
However, Brandst\"adt et~al.~\cite{BBKNP13} proved that {\sc Near-Bipartiteness} remains \NP-complete for perfect graphs.
The same authors also showed that {\sc Near-Bipartiteness} is polynomial-time solvable for $P_4$-free graphs.
\end{sloppypar}

Yang and Yuan~\cite{YY06} proved that {\sc Near-Bipartiteness} also remains \NP-complete for graphs of maximum degree~$4$.
To complement their hardness result, Yang and Yuan~\cite{YY06} showed that every connected graph of maximum degree at most~$3$ is near-bipartite except the complete graph~$K_4$ on four vertices.
This also follows from a more general result of Catlin and Lai~\cite{CL95}.
Recently we gave a linear-time algorithm for finding an independent feedback vertex set in a graph of maximum degree at most~$3$~\cite{BDFJP17}, and also proved that {\sc Near-Bipartiteness} is \NP-complete even for line graphs of maximum degree~$4$~\cite{BDFJP17b}.
It is also known that {\sc Near-Bipartiteness} is \NP-complete for planar graphs; this follows from a result of Dross, Montassier and Pinlou~\cite{DMP16}; see the arXiv version of~\cite{BDFJP17} for details.

\begin{sloppypar}
Tamura, Ito and Zhou~\cite{TIZ15} proved that {\sc Independent Feedback Vertex Set} is \NP-complete for planar bipartite graphs of maximum degree~$4$ (note that {\sc Near-Bipartiteness} is trivial for bipartite graphs).
They also proved that {\sc Independent Feedback Vertex Set} is linear-time solvable for graphs of bounded treewidth, chordal graphs and $P_4$-free graphs (the latter result generalising the result of~\cite{BBKNP13} for {\sc Near-Bipartiteness} on $P_4$-free graphs).
In~\cite{BDFJP17b} we proved that finding a minimum independent feedback vertex set is polynomial-time solvable even for $P_5$-free graphs.
We refer to~\cite{AGSS16,MPRS12} for FPT algorithms with parameter~$k$ for finding an independent feedback vertex set of size at most~$k$.
\end{sloppypar}

The {\em distance} between two vertices~$u$ and~$v$ in a graph~$G$ is the length (number of edges) of a shortest path between~$u$ and~$v$.
The {\em diameter} of a graph~$G$ is the maximum distance between any two vertices in~$G$.
In addition to their results for graphs of bounded maximum degree, Yang and Yuan~\cite{YY06} proved that {\sc Near-Bipartiteness} is polynomial-time solvable for graphs of diameter at most~$2$ and \NP-complete for graphs of diameter at most~$4$.
They asked the following question, which was also posed by Brandst\"adt et~al.~\cite{BBKNP13}:

\medskip
\noindent
{\em What is the complexity of {\sc Near-Bipartiteness} for graphs of diameter~$3$?}

\medskip
\noindent
{\bf Our Results.} We complete the complexity classifications of {\sc Near-Bipartiteness} and {\sc Independent Feedback Vertex Set} for graphs of bounded diameter.
In particular, we prove that {\sc Near-Bipartiteness} is \NP-complete for graphs of diameter~$3$, which answers the above question.
We also prove that {\sc Independent Feedback Vertex Set} is polynomial-time solvable for graphs of diameter~$2$.
This generalises the result of Yang and Yuan~\cite{YY06} for {\sc Near-Bipartiteness} restricted to graphs of diameter~$2$.

\begin{sloppypar}
\begin{theorem}\label{t-main}
Let $k\geq 0$ be an integer.
\begin{enumerate}[(i)]
\item If $k\leq 2$, then {\sc Independent Feedback Vertex Set} (and thus {\sc Near-Bipartiteness}) is polynomial-time solvable for graphs of diameter~$k$.\\[-8pt]
\item If $k\geq 3$, then {\sc Near-Bipartiteness} (and thus {\sc Independent Feedback Vertex Set}) is \NP-complete for graphs of diameter~$k$.
\end{enumerate}
\end{theorem}
\end{sloppypar}

\noindent
We prove Theorem~\ref{t-main}~(i) in Section~\ref{s-poly}.
Yang and Yuan~\cite{YY06} proved their result for {\sc Near-Bipartiteness} by giving a polynomial-time verifiable characterisation of the class of near-bipartite graphs of diameter~$2$.
We use their characterisation as the starting point for our algorithm for {\sc Independent Feedback Vertex Set}.
In fact our algorithm not only solves the decision problem but even finds a minimum independent feedback vertex set in a graph of diameter~$2$.

We prove Theorem~\ref{t-main}~(ii) in Section~\ref{s-diam3} by using a construction of Mertzios and Spirakis~\cite{MS16}, which they used to prove that {\sc $3$-Colouring} is \NP-complete for graphs of diameter~$3$.
The outline of their proof is straightforward: a reduction from {\sc $3$-Satisfiability} that constructs, for any instance~$\phi$, a graph~$H_\phi$ that is $3$-colourable if and only if~$\phi$ is satisfiable.
We reduce {\sc $3$-Satisfiability} to {\sc Near-Bipartiteness} for graphs of diameter~$3$ using the same construction, that is, we show that~$H_\phi$ is near-bipartite if and only if~$\phi$ is satisfiable.
As such, our result is an observation about the proof of Mertzios and Spirakis, but, owing to the intricacy of~$H_\phi$, this observation is non-trivial to verify.
In Section~\ref{s-diam3} we therefore repeat the construction and describe our reduction in detail, though we rely on~\cite{MS16} where possible in the proof.

\section{Independent Feedback Vertex Set for Diameter~$2$}\label{s-poly}

In this section we show how to compute a minimum independent feedback vertex set of a graph of diameter~$2$ in polynomial time.
As mentioned, our proof relies on a known characterisation of near-bipartite graphs of diameter~$2$~\cite{YY06}.
In order to explain this characterisation, we first need to introduce some terminology.

Let $G=(V,E)$ be a graph and let $X\subseteq V$.
Then the {\em $2$-neighbour set} of~$X$, denoted by~$A_X$, is the set that consists of all vertices in $V\setminus X$ that have at least two neighbours in~$X$.
A set $I\subseteq V$ is {\em independent} if no two vertices of~$I$ are adjacent.
For $u\in V$, we let $G-u$ denote the graph obtained from~$G$ after deleting the vertex~$u$ (and its incident edges).
A graph is {\em complete bipartite} if its vertex set can be partitioned into two independent sets~$S$ and~$T$ such that there is an edge between every vertex of~$S$ and every vertex of~$T$.
If~$S$ or~$T$ has size~$1$, the graph is also called a {\em star}.

\begin{theorem}[\cite{YY06}]\label{t-yy06}
A graph $G=(V,E)$ of diameter~$2$ is near-bipartite if and only if one of the following two conditions holds:
\begin{enumerate}[(i)]
\renewcommand{\theenumi}{(\roman{enumi})}
\renewcommand{\labelenumi}{(\roman{enumi})}
\item\label{cond:i}there exists a vertex~$u$ such that $G-u$ is bipartite; or
\item\label{cond:ii}there exists a set~$X$, $4\leq |X|\leq 5$, such that $(A_X,V\setminus A_X)$ is a near-bipartite decomposition.
\end{enumerate}
\end{theorem}

As noted in~\cite{YY06}, Theorem~\ref{t-yy06} can be used to solve {\sc Near-Bipartiteness} in polynomial time for graphs of diameter~$2$, as conditions~\ref{cond:i} and~\ref{cond:ii} can be checked in polynomial time.
However, Theorem~\ref{t-yy06} does not tell us how to determine the size of a minimum independent feedback vertex set.

In order to find a minimum independent feedback vertex set, we will distinguish between the two cases of Theorem~\ref{t-yy06}.
This leads to two corresponding lemmas.

\begin{lemma}\label{l-1}
Let $G=(V,E)$ be a near-bipartite graph of diameter~$2$ that contains a vertex~$u$ such that $G-u$ is bipartite.
Then it is possible to find a minimum independent feedback vertex set of~$G$ in polynomial time.
\end{lemma}

\begin{proof}
We can partition $V\setminus\{u\}$ into four independent sets $S_1$, $S_2$, $T_1$, $T_2$ (some of which might be empty) such that
\begin{enumerate}[(i)]
\renewcommand{\theenumi}{(\roman{enumi})}
\renewcommand{\labelenumi}{(\roman{enumi})}
\item\label{prop:i}$S_1\cup S_2$ and $T_1\cup T_2$ form bipartition classes of $G-u$;
\item\label{prop:ii}$u$ is adjacent to every vertex of $S_1\cup T_1$; and
\item\label{prop:iii}$u$ is non-adjacent to every vertex of $S_2\cup T_2$.
\end{enumerate}
Moreover, as~$G$ has diameter~$2$, it follows that given a vertex of~$S_2$ (respectively,~$T_2$) and a vertex of $T_1\cup T_2$ (respectively, $S_1\cup S_2$), these two vertices must either be adjacent or have a common neighbour.
As the latter is not possible, we deduce that
\begin{enumerate}[(i)]
\renewcommand{\theenumi}{(\roman{enumi})}
\renewcommand{\labelenumi}{(\roman{enumi})}
\setcounter{enumi}{3}
\item\label{prop:iv}every vertex of~$S_2$ is adjacent to every vertex of $T_1\cup T_2$, and every vertex of~$T_2$ is adjacent to every vertex of $S_1\cup S_2$ (see also \figurename~\ref{fig:STu}).
\end{enumerate}
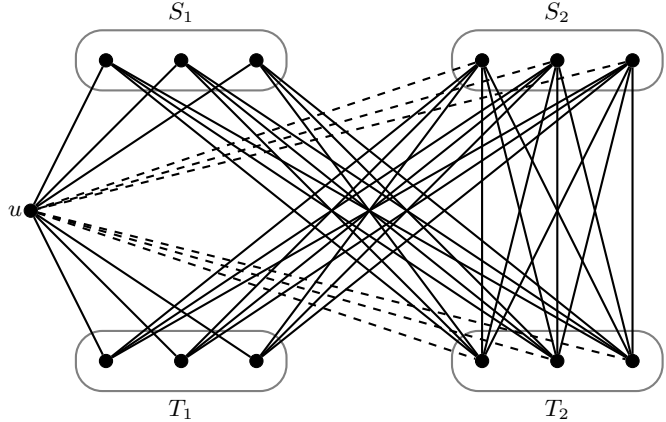
\begin{figure}[h]
\begin{center}
\tikzstyle{vertex}=[circle,draw=black, fill=black, minimum size=5pt, inner sep=1pt]
\tikzstyle{edge} =[draw,black]
\begin{tikzpicture}

\foreach \pos/\name/\number/\posl in {
(0,-2)/T/1/below,
(0,2)/S/1/above,
(5,-2)/T/2/below,
(5,2)/S/2/above}
{
\begin{scope}[shift={\pos}]
\draw[rounded corners=10pt, thick, black!50!white] (0.6,-0.4) rectangle (3.4,0.4);
\node[vertex] (\name\number1) at (1,0) {};
\node[vertex] (\name\number2) at (2,0) {};
\node[vertex] (\name\number3) at (3,0) {};
\node [\posl=0.4, align=center] at (2,0) {$\name_\number$};
\end{scope}
}

\node[vertex] (u) at (0,0) {};

\foreach \numa/\numb in {1/1,1/2,2/1,2/2,1/3,2/3,3/3,3/2,3/1} {
\path[edge, black, thick] (S1\numa) -- (T2\numb) -- (S2\numa) -- (T1\numb);
}

\foreach \num in {1,2,3} {
\path[edge, black, thick] (S1\num) -- (u) -- (T1\num);
\path[edge, black, dashed, thick] (S2\num) -- (u) -- (T2\num);
}

\node [left, align=center] at (u) {$u$};
\end{tikzpicture}
\end{center}
\caption{The graph~$G$, which consists of the vertex~$u$ and the independent sets $S_1\cup S_2$ and $T_1\cup T_2$.
Dashed lines indicate edges that are not present.
Edges between vertices of~$S_1$ and vertices of~$T_1$ are not drawn, as such edges may or may not exist.}\label{fig:STu}
\end{figure}

A (not necessarily proper) $2$-colouring of the vertices of a graph is \emph{good} if the vertices coloured~$1$ form an independent set and the vertices coloured~$2$ induce a forest.
The set of vertices coloured~$1$ in a good $2$-colouring is said to be a \emph{$1$-set} and is, by definition, an independent feedback vertex of~$G$.
A good $2$-colouring of~$G$ is {\em optimal} if its $1$-set is of minimum possible size among all good $2$-colourings.
Our algorithm colours vertices one by one with colour~$1$ or~$2$ to obtain a number of good $2$-colourings.
We will establish that our approach ensures that at least one of our good $2$-colourings is optimal.
Therefore, as our algorithm finds different good $2$-colourings, it only needs to remember the smallest $1$-set seen so far.
We note that~$G$ certainly has good $2$-colourings as, for example, we can let either $S_1\cup S_2$ or $T_1\cup T_2$ be the set of vertices coloured~$1$.

We say that an edge is a {\em $1$-edge} if both its end-points have colour~$1$ and say that a cycle of~$G$ is a {\em $2$-cycle} if all its vertices have colour~$2$.
Our algorithm will consist of a number of branches depending on the way we will colour the vertices of~$G$.
Whenever we detect a $1$-edge or a $2$-cycle in a branch, we can discard the branch as we know that we are not going to generate a good $2$-colouring.
Before we describe our algorithm, we first prove the following claim.
Here, we say that an independent set~$I$ is a {\em twin-set} if every vertex of~$I$ has the same neighbourhood.

\clm{\label{clm:1}Let~$I$ be a twin-set.
In every optimal $2$-colouring, at least $|I|-1$ vertices of~$I$ obtain the same colour.}
We prove Claim~\ref{clm:1} as follows.
If $|I|=1$, the claim is trivial.
Suppose $|I|\geq 2$ and let~$J$ be the neighbourhood of the vertices of~$I$.
Note that~$J$ is non-empty since~$|I|\geq 2$ and~$G$ is connected.
Let~$c$ be an optimal $2$-colouring of~$G$.
If~$c$ gives colour~$1$ to a vertex of~$J$, then every vertex of~$I$ must receive colour~$2$.
Now suppose that~$c$ gives colour~$2$ to every vertex of~$J$.
If $|J|=1$, then~$c$ colours every vertex of~$I$ with colour~$2$, as doing this will not create a $2$-cycle.
If $|J|\geq 2$ then, in order to avoid a $2$-cycle, at least $|I|-1$ vertices of~$I$ must be coloured~$1$.
This proves Claim~\ref{clm:1}.

\medskip
\noindent
By \ref{prop:i}, \ref{prop:iii}, \ref{prop:iv}, we find that~$S_2$ and~$T_2$ are twin-sets.
Let~$Z$ be the set of isolated vertices in the subgraph of~$G$ induced by $S_1\cup T_1$.
Then by \ref{prop:i}, \ref{prop:ii}, \ref{prop:iv}, the neighbourhood of every vertex in $Z\cap S_1$ (respectively, $Z\cap T_1$) is $T_2\cup \{u\}$ (respectively, $S_2\cup \{u\}$).
So $Z\cap S_1$ and $Z\cap T_1$ are twin-sets.

We choose one vertex from each non-empty set in $\{S_2,T_2,Z\cap S_1,Z\cap T_1\}$ and let~$W$ be the set of chosen vertices.
Note that the choice of the vertices in~$W$ can be done arbitrarily, since all four of these sets are twin-sets.
We now branch by giving all vertices in $S_2 \setminus W$ the same colour, all vertices in $T_2 \setminus W$ the same colour, all vertices in $(Z\cap S_1)\setminus W$ the same colour and all vertices in $(Z\cap T_1)\setminus W$ the same colour.
We then branch by colouring the at most four vertices of~$W$ with every possible combination of colours.
Hence the total number of branches is at most~$2^8$.
We discard any branch that yields a $1$-edge or $2$-cycle.
Let $S_1' = S_1 \setminus Z$ and $T_1'=T_1 \setminus Z$.
For each remaining branch we try to colour the remaining vertices of~$G$, which are all in $S_1'\cup T_1'\cup \{u\}$, and keep track of any minimum $1$-set found.
In the end we return a $1$-set of minimum size (recall that~$G$ has at least two $1$-sets).

For any remaining branch we do as follows.
We first give colour~$1$ to~$u$.
Then every vertex of $S_1'\cup T_1'$ must get colour~$2$.
If this does not yield a $1$-edge or $2$-cycle, we obtain a $1$-set, which we remember
if it is the smallest one found so far.

We now give colour~$2$ to~$u$.
If~$u$ was the only remaining vertex, we check for the presence of a $1$-edge or a $2$-cycle, and if none is present, we remember the $1$-set found
if it is the smallest one found so far.
Otherwise, we let $D_1,\ldots,D_r$ for some integer $r\geq 1$ be the connected components of the (bipartite) graph induced by $S_1'\cup T_1'$.
As these vertices do not belong to~$Z$, each~$D_i$ contains at least one edge.
Moreover, each~$D_i$ is bipartite.
For $i \in \{1,\ldots,r\}$, we denote the two non-empty bipartition classes of~$D_i$ by~$D_i^1$ and~$D_i^2$ such that $|V(D_i^1)|\leq |V(D_i^2)|$.
The following claim is crucial.

\clm{\label{clm:2}For $i\in\{1,\ldots,r\}$, we must either colour all vertices of~$D_i^1$ with colour~$1$ and all vertices of~$D_i^2$ with colour~$2$, or vice versa.}
We prove Claim~\ref{clm:2} as follows.
Suppose that~$D_i^1$ contains a vertex with the same colour as a vertex of~$D_i^2$.
As~$D_i$ is connected and bipartite, this means that~$D_i$ contains an edge~$vw$ whose end-vertices are either both coloured~$1$ or coloured~$2$.
In the first case, we obtain a $1$-edge.
In the second case the vertices~$u$, $v$ and~$w$ form a $2$-cycle in~$G$.
Hence we must use colours~$1$ and~$2$ for different partition classes of~$D_i$.
This proves Claim~\ref{clm:2}.

\medskip
\noindent
We now proceed as follows.
First suppose that $S_2 \cup T_2$ is non-empty.
If we coloured a vertex in~$S_2$ (respectively~$T_2$) with colour~$1$, then every vertex in~$T_1'$ (respectively~$S_1'$) must be coloured~$2$ and therefore every vertex in~$S_1'$ (respectively~$T_1'$) must be coloured~$1$ by Claim~\ref{clm:2}.
Again, in this case we discard the branch if a $1$-edge or $2$-cycle is found; otherwise we remember the corresponding $1$-set if it is the best set found so far.
In every other case, we must have coloured every vertex of non-empty set $S_2\cup T_2$ with colour~$2$.
Without loss of generality, assume that there is a vertex $s \in S_2$ that is coloured~$2$.
Then at most one vertex of~$T_1'$ may have colour~$2$, as otherwise we obtain a $2$-cycle by involving the vertices~$s$ and~$u$.
We branch by guessing this vertex and then colouring it either~$1$ or~$2$, while assigning colour~$1$ to all other vertices of~$T_1'$.
Then the only vertices with no colour yet are in~$S_1'$, but their colour is determined by the colours of the vertices in~$T_1'$ due to Claim~\ref{clm:2}.

We are left to deal with the case where $S_2\cup T_2=\emptyset$.
Claim~\ref{clm:2} tells us that we must either give every vertex of~$D_i^1$ colour~$1$ and every vertex of~$D_i^2$ colour~$2$, or vice versa.
For $i \in \{1,\ldots,r\}$ we give colour~$1$ to every vertex of every~$D_i^1$; as $|V(D_i^1)|\leq |V(D_i^2)|$, this is the best possible good $2$-colouring for this branch.

\medskip
\noindent
The correctness of our algorithm follows from the fact that we distinguish all possible cases and 
find a best possible good $2$-colouring (if one exists) in each case.
Note that it takes polynomial time to find the sets~$S_1$, $S_2$, $T_1$ and~$T_2$.
Moreover, the number of branches is~$O(n)$ and each branch can be processed in polynomial time, as we only need to search for a $1$-edge or $2$-cycle.
Hence our algorithm runs in polynomial time.\qed
\end{proof}

\begin{lemma}\label{l-2}
Let $G=(V,E)$ be a near-bipartite graph of diameter~$2$ that contains no vertex~$u$ such that $G-u$ is bipartite.
Then it is possible to find a minimum independent feedback vertex of~$G$ in polynomial time.
\end{lemma}

\begin{proof}
As~$G$ is near-bipartite, it has an independent feedback vertex set.
Let~$A$ be a minimum independent feedback vertex set.
We claim that~$G$ contains a set of vertices~$X$ of size $4\leq |X|\leq 5$ such that $A_X=A$.
This would immediately give us a polynomial-time algorithm.
Indeed, it would suffice to check, for every set~$X$ of size $4\leq |X|\leq 5$, whether $(A_X,V\setminus A_X)$ is a near-bipartite decomposition and to return a set~$A_X$ of minimum size that satisfies this condition.
This takes polynomial time.

To prove the above claim we will follow the same line of reasoning as in the proof of Theorem~\ref{t-yy06}
However, our arguments are slightly different, as we need to prove a stronger statement.

Let $B=V\setminus A$ and let~$F$ be the subgraph of~$G$ induced by~$B$.
By definition, $F$ is a forest, so all of its connected components are trees.

We will first consider the case where~$F$ has a connected component~$T$ of diameter at least~$3$.
Let~$P$ be a longest path in the tree~$T$ on vertices $v_1,\ldots,v_p$ in that order.
As~$T$ has diameter~$3$, we find that $p\geq 4$.
If $p\leq 5$, then we let $X=\{v_1,\ldots,v_p\}$.
If $p\geq 6$, then we let $X=\{v_1,v_2,v_{p-1},v_p\}$.
We will show that $A=A_X$.
Let $u\in A$.
As~$G$ has diameter~$2$ and~$A$ is an independent set, $u$ is adjacent to~$v_1$ or to a neighbour~$v^*$ of~$v_1$ in~$B$.
In the latter case, if $v^*\neq v_2$ then~$v^*$ must have a neighbour in $\{v_2,\ldots,v_p\}$, otherwise we have found a path that is longer than~$P$, but in this case~$B$ contains a cycle, a contradiction.
Hence, $u$ has at least one neighbour in $\{v_1,v_2\}$, and similarly, $u$ has at least one neighbour in $\{v_{p-1},v_p\}$.
So $A\subseteq A_X$.
Now suppose $u\in A_X$.
Note that $u\neq v_3$ due to our choice of~$X$.
Then the subgraph of~$G$ induced by $V(P)\cup \{u\}$ contains a cycle.
Hence~$u$ must belong to~$A$.
So $A_X\subseteq A$.
We conclude that $A=A_X$.

We now consider the case where every connected component of~$F$ has diameter at most~$2$.
Such components are either isolated vertices or stars (we say that the latter components are {\em star-components} and that their non-leaf vertex is the {\em star-centre};
if such a component consists of a single edge, we arbitrarily choose one of them to be the star-centre).
If~$F$ contains no star-components, then~$G$ is bipartite and therefore $G-u$ is bipartite for every vertex~$u$, a contradiction.
If~$F$ contains exactly one star-component, then by choosing~$u$ to be the star-centre we again find that $G-u$ is bipartite.
Hence~$F$ contains at least two star-components~$D_1$ and~$D_2$.
For $i=1,2$, let~$v_i$ be the star-centre and let~$w_i$ be a leaf in~$D_i$.

We choose $X=\{v_1,v_2,w_1,w_2\}$ and show that $A=A_X$.
Let $u\in A$.
As~$G$ has diameter~$2$ and~$A$ is an independent set, $u$ is either adjacent to~$w_1$ or to a neighbour of~$w_1$ in~$B$.
If this neighbour is not~$v_1$, then~$D_1$ is not a star-component, a contradiction.
Hence, $u$ has at least one neighbour in $\{v_1,w_1\}$, and similarly, $u$ has at least one neighbour in $\{v_2,w_2\}$.
So $A\subseteq A_X$.
Now suppose $u\in A_X$.
Then $X\cup \{u\}$ induces either a connected subgraph of~$G$ that contains both~$D_1$ and~$D_2$ (and is therefore not a star-component) or a subgraph with a cycle.
Hence~$u$ must belong to~$A$.
So $A_X\subseteq A$.
We conclude that $A=A_X$.
This completes the proof of our claim and thus the proof of the lemma.\qed
\end{proof}

We are now ready to prove the main result of this section.

\begin{theorem}\label{t-p}
The problem of finding a minimum independent feedback vertex set of a graph of diameter~$2$ can be solved in polynomial time.
\end{theorem}

\begin{proof}
Let~$G$ be an $n$-vertex graph of diameter~$2$.
We first check in polynomial time whether~$G$ contains a vertex~$u$ such that $G-u$ is bipartite.
If so, then we apply Lemma~\ref{l-1}.
If not, then we check in polynomial time whether~$G$ contains a set~$X$ of size $4\leq |X|\leq 5$ such that $(A_X,V\setminus A_X)$ is a near-bipartite decomposition.
If so, then~$G$ is near-bipartite and we apply Lemma~\ref{l-2}.
If not, then~$G$ is not near-bipartite due to Theorem~\ref{t-yy06}.\qed
\end{proof}

We note that the running time of the algorithm in Theorem~\ref{t-p} is determined by the time it takes to find and process each set~$X$ of size $4\leq |X|\leq 5$.
This takes~$O(n^7)$ time, as checking the existence of a set~$X$ takes~$O(n^5)$ time using brute force, determining the $2$-neighbour set~$A_X$ takes~$O(n)$ time and checking if $(A_X,V\setminus A_X)$ is a near-bipartite decomposition takes~$O(n^2)$ time.

\section{Near-Bipartiteness for Diameter~$3$}\label{s-diam3}

In this section we prove that {\sc Near-Bipartiteness} is \NP-complete for graphs of diameter~$3$.
In order to prove this, we use a construction of Mertzios and Spirakis~\cite{MS16}.
To introduce this construction, we first consider the constraint graph~$J$ defined in Figure~\ref{fig:constraint-graph}.

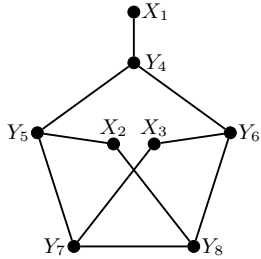
\begin{figure}[h]
\begin{center}
\scalebox{0.90}{
\tikzstyle{vertex}=[circle,draw=black, fill=black, minimum size=5pt, inner sep=1pt]
\tikzstyle{edge} =[draw,black]
\begin{tikzpicture}[scale=0.5]

   \foreach \pos/\name / \label / \posn / \dist in {
(0,4.5)/x1/$X_1$/{right}/0,
(0,3)/y1/$Y_4$/{right}/0.05,
(-0.95*3,0.31*3)/y2/$Y_5$/{left}/0,
(-0.59*3,-0.81*3)/y4/$Y_7$/{left}/0,
(0.95*3,0.31*3)/y3/$Y_6$/{right}/0,
(0.59*3,-0.81*3)/y5/$Y_8$/{right}/0,
(-0.2*3,0.2*3)/x2/$X_2$/{above}/0,
(0.2*3,0.2*3)/x3/$X_3$/{above}/0
}
       { \node[vertex] (\name) at \pos {};
       \node [\posn=\dist, align=center] at (\name) {\label};
       }
       
\foreach \source/ \dest  in {
x1/y1,y1/y2,y2/y4,y4/y5,y3/y5,y3/y1,x2/y2,x2/y5,x3/y3,x3/y4}
       \path[edge, black,  thick] (\source) --  (\dest);

\end{tikzpicture}
}
\end{center}
\caption{The constraint graph~$J$.}\label{fig:constraint-graph}
\end{figure}

\begin{lemma} \label{lem:constraint-graph}
Let~$X$ be a subset of $\{X_1, X_2, X_3\} \subset V(J)$ containing at most two vertices.
Then there exists a near-bipartite decomposition~$(A,B)$ of~$J$ such that, for $1 \leq p \leq 3$, $X_p \in A$ if and only if $X_p \in X$.
\end{lemma}
\begin{proof}
Noting the automorphic equivalence of~$X_2$ and~$X_3$, it is sufficient to consider the following two cases.
If~$X$ is a subset of $\{X_1, X_2\}$, let $A=X \cup \{Y_6,Y_7\}$.
If $X = \{X_2,X_3\}$, let $A= \{X_2,X_3,Y_4\}$.\qed
\end{proof}
\begin{sloppypar}
Notice that there is no near-bipartite decomposition of~$J$ with $\{X_1,X_2,X_3\} \subseteq A$.
Combined with the above lemma, this gives an idea of how this will be used later.
The vertices~$X_1$,~$X_2$ and~$X_3$ will represent literals in a clause of an instance of {\sc $3$-Sat} and membership of~$A$ will indicate that a literal is false: thus~$A$ can be extended to a near-bipartite decomposition except when every literal is false.
(In~\cite{MS16}, a weaker result was shown: one can always find a $3$-colouring of~$J$ such that members of a chosen \emph{proper} subset of $\{X_1, X_2, X_3\}$ belong to the same class and excluded members do not belong to that class.)
\end{sloppypar}

Let~$\phi$ be an instance of {\sc $3$-Sat} with~$m$ clauses $C_1, \ldots, C_m$ and~$n$ variables $v_1,\ldots,v_n$.
We may assume that each clause has three distinct literals.
For a clause~$C_k$ in~$\phi$, we describe a \emph{clause graph}~${\mathcal C}^k$, illustrated within Figure~\ref{fig:hphi}.
We think of~${\mathcal C}^k$ as an array of $n+5m+1$ rows and eight columns.
In each row except the last, every (row,column) position contains exactly two vertices, which we refer to as the \emph{true vertex} and the \emph{false vertex}, and we say that these two vertices are \emph{mates}.
The first~$n$ rows form the \emph{variable block} of the graph and we think of row~$i$ as representing the variable~$v_i$.
The next~$5m$ rows are made up of~$m$ \emph{clause blocks}~${\mathcal C}^{k,1}, {\mathcal C}^{k,2}, \ldots, {\mathcal C}^{k,m}$, each of five rows.
Every true vertex of the variable and clause blocks is joined by an edge to every false vertex in the same row except its mate.
Hence the vertices of each row induce a complete bipartite graph minus a matching.
In the final row, each column contains a single vertex, and each of these vertices is joined by an edge to every other vertex in the same column.
We call this row the \emph{dominating block}.
We complete the definition of the clause graph by describing how we add further edges so that it contains the constraint graph~$J$ as an induced subgraph.
Let the literals of~$C_k$ be $x_{\ell_1},x_{\ell_2},x_{\ell_3}$.
We choose vertices from the first three columns of the variable block of~${\mathcal C}^k$ that we will denote $X_1^k, X_2^k, X_3^k$ to represent the literals.
If~$x_{\ell_p}$ is the variable~$v_i$, then we choose as~$X_p^k$ a vertex from row~$i$ and column~$p$, and choose the true vertex if the literal is positive and the false vertex if the literal is a negated variable.
For $p \in \{4,\ldots,8\}$, let~$Y_p^k$ be the true vertex from the $(p-3)$th row and $p$th column of the clause block~${\mathcal C}^{k,k}$.
Finally add the ten edges
$\{ X_1^kY_4^k, X_2^kY_5^k, X_2^kY_8^k, X_3^kY_6^k, X_3^kY_7^k, Y_4^kY_5^k, Y_4^kY_6^k, Y_5^kY_7^k, Y_6^kY_8^k, Y_7^kY_8^k\}$
so that $\{X_1^k,X_2^k,X_3^k, Y_4^k, Y_5^k,Y_6^k,Y_7^k,Y_6^k\}$ induces the constraint graph~$J$.

\usetikzlibrary{decorations.pathreplacing,angles,quotes}

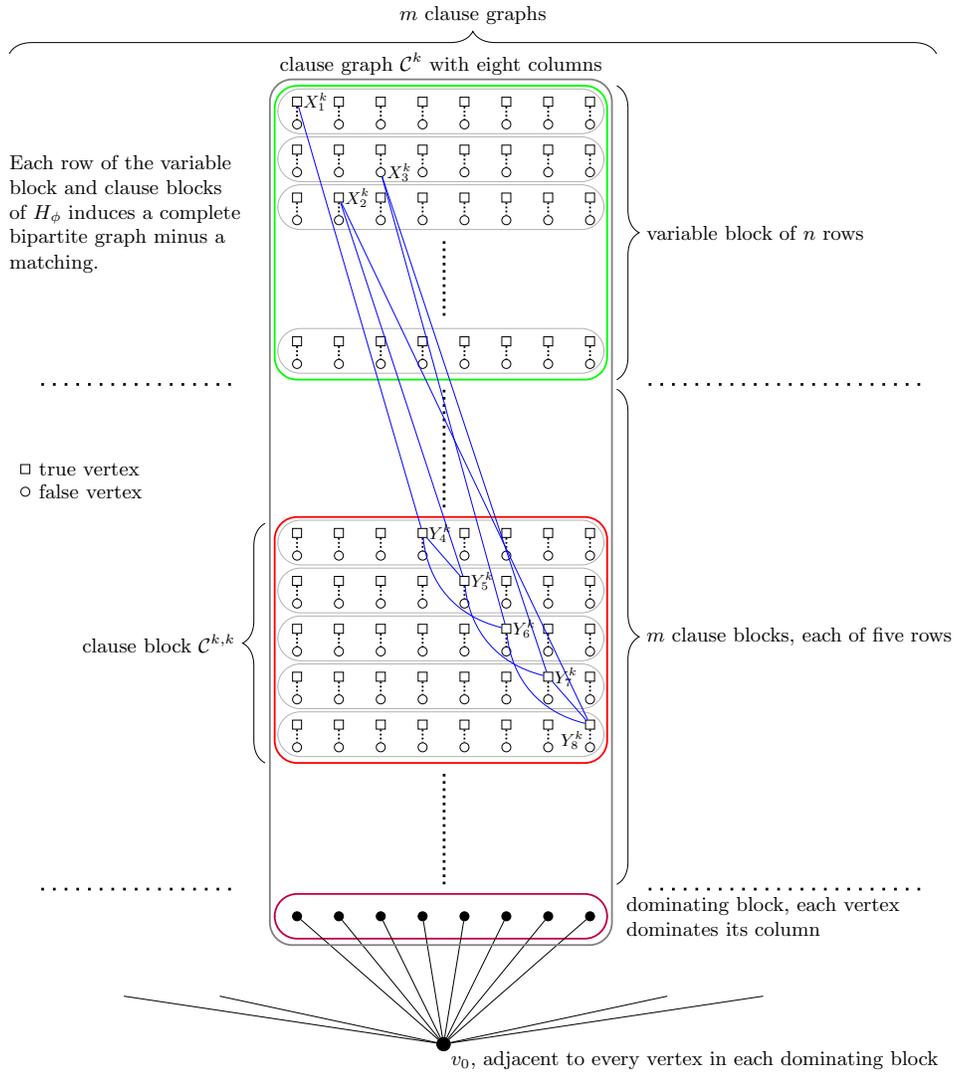
\begin{figure}
\begin{center}
\tikzstyle{vertex}=[circle,draw=black, fill=black, minimum size=5pt, inner sep=1pt]
\tikzstyle{edge} =[draw,black]
\scalebox{0.85}{

\begin{tikzpicture}[scale=0.5]

\begin{scope}[yshift=15cm]
\foreach \y in {2,5,6,7}{
\node[draw, rectangle, minimum size=4pt, inner sep=0pt] (a1\y) at (3,1.5*\y) {};
\node[draw, circle, minimum size=4pt, inner sep=0pt, below =of a1\y, yshift=0.8cm] (b1\y) {};
\path[edge, black,  thick, densely dotted] (a1\y) --  (b1\y);
\foreach \x [remember=\x as \lastx (initially 1)] in {2,...,8}{
    \node [draw, rectangle, minimum size=4pt, inner sep=0pt, right =of a\lastx\y, xshift=-0.5cm] (a\x\y) {};
    \node [draw, circle, minimum size=4pt, inner sep=0pt, right =of b\lastx\y, xshift=-0.5cm] (b\x\y) {};
       \path[edge, black,  thick, densely dotted] (a\x\y) --  (b\x\y);
}
\draw[rounded corners=10pt, thin, black!30!white] (2.4,1.5*\y-1) rectangle (12.6,1.5*\y+0.4);
}
\path[dotted, draw, very thick] (7.6,3.8) -- (7.6,6.2);
\draw[rounded corners=10pt, thick, green] (2.3,1.8) rectangle (12.7,11);
\draw[decoration={brace,mirror,amplitude=10pt},decorate]
  (13,1.8) -- node[right=10pt] {variable block of~$n$ rows} (13,11);
\end{scope}

\begin{scope}[yshift=4.5cm]
\foreach \y in {1,...,5}{
\node[draw, rectangle, minimum size=4pt, inner sep=0pt] (t1\y) at (3,1.5*\y) {};
\node[draw, circle, minimum size=4pt, inner sep=0pt, below =of t1\y, yshift=0.8cm] (f1\y) {};
\path[edge, black,  thick, densely dotted] (t1\y) --  (f1\y);
\foreach \x [remember=\x as \lastx (initially 1)] in {2,...,8}{
    \node [draw, rectangle, minimum size=4pt, inner sep=0pt, right =of t\lastx\y, xshift=-0.5cm] (t\x\y) {};
    \node [draw, circle, minimum size=4pt, inner sep=0pt, right =of f\lastx\y, xshift=-0.5cm] (f\x\y) {};
       \path[edge, black,  thick, densely dotted] (t\x\y) --  (f\x\y);
}
\draw[rounded corners=10pt, thin, black!30!white] (2.4,1.5*\y-1) rectangle (12.6,1.5*\y+0.4);
}
\draw[rounded corners=10pt, thick, red] (2.3,0.3) rectangle (12.7,8);
\draw[decoration={brace,amplitude=10pt},decorate]
  (2,0.3) -- node[left=10pt] {clause block~${\mathcal C}^{k,k}$} (2,7.8);
\end{scope}

\path[dotted, draw, very thick] (7.6,12.8) -- (7.6,16.6);
\path[dotted, draw, very thick] (7.6,1) -- (7.6,4.5);
\draw[decoration={brace,mirror,amplitude=10pt},decorate]
  (13,1) -- node[right=10pt] {$m$ clause blocks, each of five rows} (13,16.5);

\begin{scope}[yshift=0cm]
\node[draw, circle, fill=black, minimum size=4pt, inner sep=0pt] (d1) at (3,0) {};
\foreach \x [remember=\x as \lastx (initially 1)] in {2,...,8}{
    \node [draw, circle, fill=black, minimum size=4pt, inner sep=0pt, right =of d\lastx, xshift=-0.5cm] (d\x) {};
}
\draw[rounded corners=10pt, thick, purple] (2.3,-0.7) rectangle (12.7,0.7);
\node [text width=5cm] at (18.3,0) {dominating block, each vertex dominates its column};

\draw[decoration={brace,amplitude=10pt},decorate]
  (-6,27) -- node[above=10pt] {$m$ clause graphs} (23,27);

    \node [draw, circle, fill=black, minimum size=6pt, inner sep=0pt, yshift=-2cm] (v0) at ($(d1)!0.50!(d8)$) {};
       \node [below right=0cm, align=center] at (v0) {$v_0$, adjacent to every vertex in each dominating block};

\foreach \dest  in {d1,d2,d3,d4,d5,d6,d7,d8}
       \path[edge, black] (v0) --  (\dest);

\foreach \x in {-10,-7,7,10}
       \path[edge, black] (v0) -- ($(v0)+(\x,1.5)$);

\end{scope}

\draw[rounded corners=10pt, thick, black!50!white] (2.15,-0.9) rectangle (12.85,26.2);
\node at (7.5,26.7) {clause graph~${\mathcal C}^k$ with eight columns};

\path[loosely dotted, draw, very thick] (14,0.85) -- (22.5,0.85);
\path[loosely dotted, draw, very thick] (-5,0.85) -- (1,0.85);
\path[loosely dotted, draw, very thick] (14,16.65) -- (22.5,16.65);
\path[loosely dotted, draw, very thick] (-5,16.65) -- (1,16.65);

\foreach \source/ \dest  in {
a17/t45,t45/t54,t72/t81,a25/t54,a25/t81,b36/t63,b36/t72}
       \path[edge, blue] (\source) --  (\dest);
\path [edge, blue] (t45) to [bend right=35] (t63);  
\path [edge, blue] (t54) to [bend right=35] (t72);  
\path [edge, blue] (t63) to [bend right=35] (t81);  

\node [right=0.0cm, scale=0.85, align=center] at (a17) {\footnotesize{$X_1^k$}};
\node [right=0.0cm, scale=0.85, align=center] at (a25) {\footnotesize{$X_2^k$}};
\node [right=0.0cm, scale=0.85, align=center] at (b36) {\footnotesize{$X_3^k$}};
\node [right=0.0cm, scale=0.85, align=center] at (t45) {\footnotesize{$Y_4^k$}};
\node [right=0.0cm, scale=0.85, align=center] at (t54) {\footnotesize{$Y_5^k$}};
\node [right=0.0cm, scale=0.85, align=center] at (t63) {\footnotesize{$Y_6^k$}};
\node [right=0.0cm, scale=0.85, align=center] at (t72) {\footnotesize{$Y_7^k$}};
\node [below left=0.0cm, scale=0.85, align=center] at (t81) {\footnotesize{$Y_8^k$}};

\begin{scope}[yshift=21cm,xshift=19.5cm]
\node[draw, rectangle, minimum size=4pt, inner sep=0pt] (true) at (-25,-7) {};
\node[draw, circle, minimum size=4pt, inner sep=0pt, below =of true, yshift=0.8cm] (false) {};
\node [right=0.1cm, align=center] at (true) {true vertex};
\node [right=0.1cm, align=center] at (false) {false vertex};
\end{scope}

\node [text width=4cm] at (-2,22) {Each row of the variable block and clause blocks of~$H_\phi$ induces a complete bipartite graph minus a matching.};

\end{tikzpicture}
}
\end{center}
\caption{The graph~$H_\phi$ with the focus on one of the constituent clause graphs~${\mathcal C}^k$, where $C_k=\{v_1 \vee v_3 \vee \overline{v_2}\}$.
The blue edges illustrate the induced constraint graph.}\label{fig:hphi}
\end{figure}

We now define the graph~$H_\phi$.
It contains:
\begin{itemize}
\item the disjoint union of clause graphs~${\mathcal C}^k$, $1 \leq k \leq m$ (we think of the clause graphs as being arranged side-by-side, so that they form an array of $n+5m+1$ rows and~$8m$ columns),
\item edges from each true vertex of each clause graph to each false vertex in the same row of other clause graphs, and
\item an additional vertex~$v_0$ joined to each vertex in the dominating block of each clause graph.
\end{itemize}

Note that each column of~$H_\phi$ contains exactly one vertex that is in a constraint graph~$J$ and the only rows that contain more than one such vertex are those in the variable block.

\begin{lemma}[\protect{\cite[Lemma~2]{MS16}}]
For an instance~$\phi$ of {\sc $3$-Sat}, $H_\phi$ has diameter~$3$.
\end{lemma}

Note that in~\cite{MS16}, Lemma~2 proves the bound on the diameter for a graph that is a spanning subgraph of~$H_\phi$ which is, of course, sufficient for an upper bound for the diameter of~$H_\phi$ and it is easy to see that the diameter is not less than~$3$.
We note also that~$H_\phi$ does not contain any triangles or any vertices that are siblings (two vertices are siblings if the neighbourhood of one is a subset of the neighbourhood of the other) so {\sc Near-Bipartiteness} is also \NP-complete for such instances.

\begin{theorem}\label{thm:diam3}
{\sc Near-Bipartiteness} is \NP-complete for graphs of diameter at most~$3$.
\end{theorem}

\begin{proof}
We prove that {\sc $3$-Sat} can be polynomially-reduced to {\sc Near-Bipartiteness} by showing that~$\phi$ is satisfiable if and only if~$H_\phi$ has a near-bipartite decomposition~$(A,B)$.

\medskip
\noindent
($\Rightarrow$) Suppose that~$\phi$ has a satisfying assignment.
Let~$v_0$ be in~$A$, and let the vertices of all the dominating blocks be in~$B$.
If the variable~$v_i$ is true, then let~$B$ contain all the true vertices of row~$i$ of the variable blocks of each clause graph.
Otherwise let~$B$ contain the false vertices.
In each case, let~$A$ contain the mates of these vertices.
Consider the constraint graph that is an induced subgraph of each clause graph.
The vertices~$X_1$, $X_2$ and~$X_3$ have been assigned to either~$A$ or~$B$ with at most two, representing false literals, belonging to~$A$.
By Lemma~\ref{lem:constraint-graph}, we can assign the remaining vertices of the subgraph (which are all true vertices of clause blocks) to~$A$ and~$B$ such that on the subgraph they form a near-bipartite decomposition.
When we assign a true vertex of a clause block to~$A$ or~$B$, we assign all other true vertices in the same row of~$H_\phi$ to the same set and assign their mates to the other set.
As each row of the clause blocks contains only one vertex in a constraint graph, this process assigns every vertex in~$H_\phi$ to exactly one of~$A$ and~$B$, and we have assigned every vertex of~$H_\phi$ to~$A$ or~$B$.

It is immediately clear that~$A$ is an independent set.
We must show that~$B$ contains no cycles.
We know that~$B$ contains all the vertices of the dominating blocks and, in each row, either all the true vertices or all the false vertices.
Thus if~$B$ contains a cycle then all the vertices of the cycle belong to the same clause graph (the only edges going between distinct clause graphs are those joining true vertices to false vertices in the same row).
Let~$G_B$ be a subgraph of a clause graph induced by vertices of~$B$.
Then each true and false vertex not in the constraint graph has degree~$1$ (due to the edge joining it to the dominating block), and each vertex in the dominating block has at most one neighbour with degree more than~$1$ (since it only has one neighbour in the constraint graph).
Thus if~$G_B$ contains a cycle then it belongs to the constraint graph, contradicting how~$A$ and~$B$ were chosen.

\medskip
\noindent
($\Leftarrow$) Suppose~$A$ and~$B$ form a near-bipartite decomposition of~$H_\phi$.
Then~$B$ can be decomposed into two independent sets, and these, along with~$A$, can be considered a $3$-colouring.
In~\cite[Theorem~5]{MS16}, it is shown that if~$H_\phi$ has a $3$-colouring, then~$\phi$ is satisfiable.\qed
\end{proof}

\section{Conclusions}\label{s-con}

We completed the computational complexity classifications of {\sc Near-Bipartiteness} and {\sc Independent Feedback Vertex Set} for graphs of diameter~$k$ for every integer $k\geq 0$.
We showed that the complexity of both problems jumps from being polynomial-time solvable to \NP-complete when~$k$ changes from~$2$ to~$3$.

We recall that near-bipartite graphs are $3$-colourable.
Interestingly, the complexity of {\sc $3$-Colouring} for graphs of diameter~$k$ has not yet been settled, as there is one remaining case left, namely when $k=2$.
This is a notorious open problem, which has been frequently posed in the literature (see, for example,~\cite{BKM12,BFGP13,MS16,Pa15}).
We note that the approach of solving {\sc Near-Bipartiteness} and {\sc Independent Feedback Vertex Set} for graphs of diameter~$2$ does not work for {\sc $3$-Colouring}.
For instance, we cannot bound the size of the set~$X$ in Lemma~\ref{l-2} if we drop the condition that the union of two colour classes must induce a forest.

\bibliography{mybib}

\end{document}